\documentclass[11pt,a4paper]{article}
\usepackage{float}
\usepackage{graphicx}
\usepackage{amssymb}
\usepackage{color}
\usepackage{algorithm}
\usepackage{algpseudocode}
\algtext*{EndIf}
\algtext*{EndFor}
\usepackage{verbatim}
\usepackage{subfig}
\usepackage{amsmath}
\usepackage{amsthm}

\newtheorem{definition}{Definition}
\newtheorem{property}{Property}
\newtheorem{lemma}{Lemma}
\newtheorem{theorem}{Theorem}
\newtheorem{corollary}{Corollary}
\newtheorem{remark}{Remark}

\newcommand*{\Reals}{\mathbb{R}}%
\newcommand*{\AbsFamily}{{F}}
\newcommand*{\InpSet}{{F}}

\newcommand*{\nclus}{k}

\newcommand{\aPoint}{t}
\newcommand{\rep}{t}
\newcommand*{\site}{s}
\newcommand{\SiteSet}{S}
\newcommand{\nlevs}{h}
\newcommand{\Plane}{\mathbb{R}^2}
\newcommand{\aClus}{C}
\newcommand*{\curClus}{C}
\newcommand{\pointFromClus}{c}
\newcommand{\aClusI}{P}

\newcommand{\aClusIl}{Q}
\newcommand{\df}{d_{\text{\textup{\textrm{f}}}}}
\newcommand*{\FVD}{\text{\textup{\textrm{\textmd{FVD}}}}}

\newcommand*{\freg}{\text{\textup{\textrm{\textmd{freg}}}}}
\newcommand*{\fskel}{\mathcal{T}}
\newcommand*{\hreg}{\text{\textup{\textrm{\textmd{hreg}}}}}

\newcommand*{\HVD}{\text{\textup{\textrm{\textmd{HVD}}}}}
\newcommand*{\VD}{\text{\textup{\textrm{\textmd{VD}}}}}
\newcommand*{\hvd}{\text{\textup{\textrm{\textmd{HVD}}}}}

\newcommand*{\cand}{\text{\textup{\textrm{cand}}}}
\DeclareMathOperator{\bd}{\partial}
\DeclareMathOperator{\conv}{conv}
\newcommand*{\cardinality}[1]{\lvert#1\rvert}

\newcommand*{\Sep}{\text{\textup{\textrm{CD}}}}
\newcommand*{\complexity}[1]{\lVert#1\rVert}
\newcommand*{\E}{\mathbb{E}}

\newcommand*{\cdisk}{\mathcal{D}}
\newcommand*{\cldf} {D^{\text{\textup{\textrm{\textmd{f}}}}}}

\floatname{algorithm}{Procedure}

\newcommand*{\Kpure}{K^{\text{\textup{\textrm{\textmd{pure}}}}}}
\newcommand*{\Npure}{N^{\text{\textup{\textrm{\textmd{pure}}}}}}
\newcommand*{\Kmix} {K^{\text{\textup{\textrm{\textmd{mix}}}}}}
\newcommand*{\Nmix} {N^{\text{\textup{\textrm{\textmd{mix}}}}}}
\newcommand*{\cpure}{c^{\text{\textup{\textrm{\textmd{pure}}}}}}
\newcommand*{\cmix} {c^{\text{\textup{\textrm{\textmd{mix}}}}}}

\newcommand*{\hregl}{\smash{\text{\textup{\textrm{hreg}}}}_{F}^{(\ell)}}
\newcommand*{\hreglc}{\smash{\overline{\text{\textup{\textrm{hreg}}}}_{F^{(\ell)}}}}

\newcommand{\deleted}[1]{}
\newcommand*{\newc}{C}

\begin{document}
\title{A Randomized Incremental Algorithm for the Hausdorff Voronoi Diagram
         of Non-crossing Clusters
\thanks{%
       Supported in part by the Swiss National Science Foundation
       project 20GG21-134355, under the auspices of the ESF EUROCORES program
       EuroGIGA/VORONOI.%
       }%
}
\date{}

\author{
  Panagiotis Cheilaris\thanks{Faculty of Informatics,
Universit{\`a} della Svizzera \mbox{italiana}, Lugano, Switzerland}   \and
  Elena Khramtcova\footnotemark[2]   \and 
  Stefan Langerman\footnote{D{\'e}partment d'Informatique, Universit{\'e} Libre de Bruxelles, Brussels, Belgium} 
  \and
  Evanthia Papadopoulou\footnotemark[2] 
      }



\maketitle

\begin{abstract}

In the Hausdorff Voronoi diagram of a 
family of \emph{clusters of points}
in the plane, the distance between a point $t$ and a cluster $P$ is
measured as the maximum distance between
$t$ and any point in $P$, and the diagram is defined 
in a nearest-neighbor sense for the input clusters.
In this paper we consider 
\emph{non-crossing} 
clusters in the plane, for which the  combinatorial
complexity of the Hausdorff Voronoi diagram is linear in
the total number of points, $n$, on the convex hulls of all clusters. 
We present a randomized incremental construction, 
based on point location,  that computes  
this diagram in expected $O(n\log^2{n})$ time and expected $O(n)$ space.
Our techniques efficiently handle non-standard characteristics of generalized
Voronoi diagrams, such as sites of non-constant complexity, 
sites that are not enclosed in their Voronoi regions, and empty Voronoi
regions.  The diagram finds direct applications in VLSI computer-aided design.  

\end{abstract}

\section{Introduction} 
\label{sec:intro}

Given a set of simple \emph{sites} contained in some space, 
the \emph{Voronoi region} of each site $s$ is the geometric locus
of points in this 
space that are closer to $s$ than to any other site.
In the classic Voronoi diagram, each site is a point, and closeness is
measured according to the Euclidean distance.
In this work, we investigate randomized algorithms for constructing the
\emph{Hausdorff Voronoi diagram}.
The containing space is $\Reals^2$,
each site is a cluster of points (i.e., a set of points), and
closeness of a point $t$ 
to a cluster $P$
is measured by the \emph{farthest distance}
$\df(t, P) = \max_{p \in P}d(t,p)$, where $d(\cdot,\cdot)$ denotes 
the 
Euclidean distance between two points.
%
The farthest distance equals the \emph{Hausdorff distance}
between $t$ and  $P$, hence, the name of the diagram.
%
No two different clusters may have a common point.  
Let $\nclus$  denote the number of input clusters and let $n$ be the total number
of points on the convex hulls of all clusters. 
\deleted{ 
Let $\nclus$  denote the number of input clusters and let $n$ be the total number
of points on the convex hulls of all clusters. 
No two different clusters may have a common point.  
}
As it will be evident in the sequel, only points on the convex hulls
of individual clusters can be  
relevant to the Hausdorff diagram, thus, we
assume that each cluster equals its convex hull. 

Our motivation for investigating the Hausdorff Voronoi diagram comes
from applications in Very Large Scale Integration (VLSI) circuit design. The Hausdorff
Voronoi diagram has been used to estimate efficiently the \emph{critical
area} of a chip design for various open faults
\cite{Papadopoulou2001ieeecad,ep2004algorithmica,Papadopoulou2011ieeecad}; 
critical area is a measure reflecting the
sensitivity of a VLSI design to random defects during manufacturing.
The diagram can also find applications in other networks embedded in
the plane, such as transportation networks, 
where critical area may need to be extracted for the purpose of
flow control and  disaster avoidance. 

\subsection{Previous Work}

The Hausdorff Voronoi diagram was first considered by Edelsbrunner~{et~al.}~\cite{EGS1989}
under the name \emph{cluster Voronoi diagram}.
For arbitrary clusters, the authors proved that the combinatorial
complexity of the diagram is $O(n^2 \alpha(n))$ and also provided
an algorithm of the same time complexity for its construction,
where $\alpha(n)$ is the inverse Ackermann function.
These bounds were later improved to $O(n^2)$ \cite{EP-HVD-d_and_q}.
The Hausdorff Voronoi diagram is equivalent to
the upper envelope of a family of lower envelopes
of an arrangement of hyperplanes in $\mathbb{R}^3$ 
(each envelope corresponds to a cluster)~\cite{EGS1989}.
When the convex hulls of the clusters are disjoint,
the combinatorial complexity of the diagram is $O(n)$. 
In fact, the diagram remains linear for  \emph{non-crossing}
clusters  (see Definition~\ref{def:noncrossing}),
a weaker condition than disjointness of convex hulls~\cite{EP-HVD-d_and_q}.
The $O(n^2)$-time  algorithm to compute the Hausdorff 
diagram, 
although optimal in the worst case, remains quadratic in all cases, even for 
linear-complexity  instances of the diagram.

\begin{definition}
\label{def:noncrossing}
Two clusters $\aClusI$ and $\aClusIl$ are called \emph{non-crossing} if
the convex hull of $\aClusI \cup \aClusIl$
admits at most two supporting segments with one endpoint in $\aClusI$
and one endpoint in $\aClusIl$. See Fig.~\ref{fig:disjointcrossing}.
\end{definition}

\begin{figure} [hbtp]
\centering
\includegraphics%
{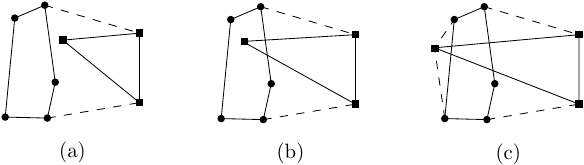} %
\caption{%
Non-crossing ((a) disjoint and (b) non-disjoint) clusters,  and (c) crossing clusters.
Dashed lines show the 
relevant supporting segments.}
\label{fig:disjointcrossing}%
\end{figure}

The combinatorial complexity (size) of the Hausdorff Voronoi diagram was  
shown to be $O(n+m)$, where $m$ is the number of certain
supporting segments between pairs of crossing clusters, called
\emph{crucial}, and this was shown to be tight~\cite{EP-HVD-d_and_q,ep2004algorithmica}.
In the worst case, $m$ is $O(n^2)$, however,
if clusters are 
non-crossing
($m=0$) the diagram  has size $O(n)$. There are plane sweep and divide and conquer
algorithms for constructing the Hausdorff Voronoi diagram of arbitrary
clusters~\cite{ep2004algorithmica,EP-HVD-d_and_q}.
Both algorithms have a $K \log{n}$ term in their time complexity, where $K$
is a parameter reflecting the number of pairs of clusters such that one is contained
in a specially defined enclosing circle of the other,
for example, the minimum enclosing circle 
in~\cite{EP-HVD-d_and_q}. 
However, $K$ can be $\omega(n)$ (superlinear), 
even in the case of non-crossing clusters
when the combinatorial
complexity of the diagram is $O(n)$.

A more recent parallel algorithm~\cite{Dehne-coarse}
constructs the Hausdorff Voronoi diagram of non-crossing clusters
in $O(p^{-1} n \log^4 n)$ time with $p$ processors, which implies a divide and conquer sequential
algorithm of time complexity $O(n \log^4 n)$ and space complexity
$O(n\log^2 n)$.

\deleted{
For non-crossing clusters, Voronoi regions are
connected and the combinatorial
complexity of the Hausdorff diagram is 
thus
$O(n)$. 
}
The Hausdorff Voronoi diagram of non-crossing clusters
is an instance of  
\emph{abstract Voronoi diagrams},  introduced by
Klein~\cite{klein-habilitation}. 
Using the randomized incremental 
construction for abstract Voronoi diagrams~\cite{klein-avd},
the Hausdorff diagram for non-crossing sites 
can be computed in expected $O(bn \log n)$ time,
where $b$ is the time 
to compute the bisector between two clusters~\cite{ahknu1997dcg};
however, $b$ can be $\Theta(n)$. 
This
framework was successfully applied
to compute the Voronoi diagram of disjoint polygons~\cite{snoyeink} in $O(k\log{n})$ time,
where $k$ is the number of sites, and $n$ is their total
combinatorial complexity. However, it is not easy
to apply a similar approach to the Hausdorff
Voronoi diagram,  because of a fundamental difference between
the farthest and the nearest distance between a point and a convex
polygon~\cite{edels-extreme}.

The  Hausdorff Voronoi diagram is a \emph{min-max} 
type of diagram, where every point in the plane is assigned to the region of the 
\emph{nearest} cluster with respect to the \emph{farthest} distance.
A ``dual'' \emph{max-min} diagram
has also been considered in the literature
\cite{ahiklmps2001ewcg,fpvd2011cg,hks1993dcg}, having been termed the \emph{farthest color Voronoi diagram}.
If clusters are disjoint simple polygons,
a divide and conquer algorithm computes this diagram in 
$O(n \log^3 n)$ time~\cite{fpvd2011cg}, where $n$ is the total 
complexity of the polygons.

For more information on generalized Voronoi diagrams see, e.g., the
book of Aurenhammer et al.~\cite{Aurenbook}.

\subsection{Our Contribution}
\label{subsec:our_contrib}
In this paper we give a randomized incremental 
approach to compute the Hausdorff Voronoi diagram of a family of
$\nclus$ non-crossing clusters, based on  point location. 
Non-crossing clusters is the condition under which 
Voronoi regions are connected and 
%
the combinatorial
complexity of the Hausdorff diagram is $O(n)$. 
In addition, it is of interest to our motivating application, where the
number of crossing clusters is typically small and can possibly be
regarded as zero.

In our algorithm, clusters are inserted in random order one by one,
while the diagram is maintained in a  dynamic data
structure,
which can answer various types of  point location queries efficiently.
 To insert a cluster, 
a \emph{representative point} in the new Voronoi region 
is first identified, and then the new region is traced around it, 
while the data structure is updated. This general 
technique has been followed previously by  
randomized algorithms to construct the Voronoi diagram  of convex
objects and the Delaunay
triangulation~\mbox{\cite{Boissonat-curved_vd,devillers-delaunay_hier,KARAVELAS-vd_of_co_in_the_plane}}. 
Identifying a representative point in the  new Hausdorff Voronoi
region is a major technical challenge
in the construction of the Hausdorff diagram.
This is difficult for the Hausdorff diagram
because: (a)~the region of the new cluster might not contain
any of its points, (b)~clusters have non-constant  complexity,
and (c)~the addition of a new cluster may make an existing region empty. 

The dynamic data 
structure that we use is a variant of 
the \emph{Voronoi hierarchy}~\cite{KARAVELAS-vd_of_co_in_the_plane}, 
which in turn is based on the Delaunay  
hierarchy~\cite{devillers-delaunay_hier}, and  which we augment so
that it can efficiently handle the difficulties listed above. 
Our  augmentation of the Voronoi hierarchy may be of interest
to 
incremental constructions of other non-standard types of 
generalized Voronoi diagrams. The expected running time of our algorithm is 
$O(n\log{n}\log{\nclus})$ and the expected space complexity is $O(n)$.
To achieve this time complexity, we also exploit a technique 
by Aronov {et al.}~\cite{aronov06data} to efficiently 
query the static farthest Voronoi diagram of a  given cluster. 

Our algorithm can also be implemented in \emph{deterministic} $O(n)$ space
and $O(n\log^2{n}(\log\log{n})^2)$ expected running time,
using the dynamic point location data structure by
Baumgarten~{et al.}~\cite{Baumgarten}.

In a companion paper~\cite{KP14} we also provide a randomized
incremental construction for the Hausdorff Voronoi diagram, which
avoids point location on a dynamic  data structure, but instead  
maintains a \emph{conflict} or a \emph{history graph}~\cite{Clarkson_rand_sampling_2}. 
The time complexity of both methods is comparable and they can be regarded complementary; 
the choice may simply depend on the availability of an
already existing framework. For example, the 
Delaunay hierarchy is already available in the  CGAL library,\footnote{www.cgal.org} 
therefore, the randomized incremental construction of this paper, based on point location on
such a hierarchical data structure, could
be the method of choice in such an environment.

This paper is organized as follows.
In Section~\ref{sec:prelim} we introduce notation and review known
properties of the Hausdorff Voronoi diagram. In
Section~\ref{sec:algorithm} we describe our randomized incremental construction algorithm. In
Sections~\ref{sec:sep_decomp} and~\ref{sec:vh} we discuss the data
structures to perform various types of point location queries.
Section~\ref{sec:tracing}  provides details for the  tracing of a new Voronoi  region. 
In Section~\ref{sec:compl_analysis} we analyze the complexity of our algorithm. 
We conclude with a brief discussion in Section~\ref{sec:discus}.

\section{Preliminaries}
\label{sec:prelim}

\begin{figure}
\centering
\includegraphics{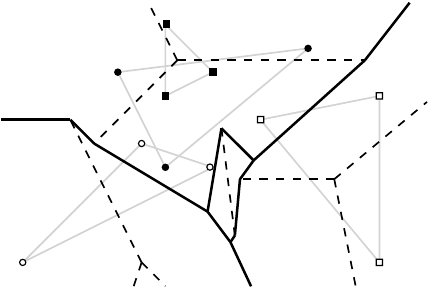}
\caption{The HVD of four clusters}
\label{fig:hvd}
\end{figure}

Let $\AbsFamily$ be a family of  $\nclus$
non-crossing 
clusters of points in the plane such that no 
two clusters have a common point, and let $n=|\cup\AbsFamily|$.
Unless stated otherwise,
we use  $\aClus$ to denote a cluster in $\AbsFamily$, and  $c$ to denote a
point within $C$. Let $\conv C$ denote the convex hull of $C$.
For simplicity of presentation, we follow a general position assumption  that  
no four points  lie on the same circle.
This assumption can be removed similarly to an ordinary Voronoi
diagram of points, e.g., following techniques of symbolic
perturbation~\cite{Seidel}.

The \emph{farthest Voronoi diagram} of $\aClus$, in brief $\FVD(\aClus)$, 
is a partitioning of the plane into regions such that the 
\emph{farthest Voronoi region} of a point $c\in \aClus$~is 
\[\freg_{\aClus}(\pointFromClus)  = \{\aPoint \mid
\forall c' \in \aClus \setminus\{\pointFromClus\} \colon
d(\aPoint,\pointFromClus) > d(\aPoint,c') \}.\]

\noindent
It is well known that $\freg_{\aClus}(c)\neq \emptyset$, if
and only if $c$ is a vertex of $\conv C$.
It is also well known the graph structure of $\FVD(\aClus)$ is
a tree.

Let $\fskel(C) =  \Plane \setminus \bigcup_{c \in C} \freg_C(c)$, if
$|C|>1$;
and let $\fskel(C) = c$, if $C = \{c\}$.
For $|C|>1$, $\fskel(C)$ is a tree corresponding to 
the graph structure of $\FVD(\aClus)$.
We assume that 
$\fskel(C)$ is rooted at a point at infinity 
along an unbounded Voronoi edge.

The \emph{Hausdorff Voronoi region} of 
a cluster $C \in \AbsFamily$ 
is defined as:
\[\hreg_{\AbsFamily}(C)  = \{p \mid \forall C'\in \AbsFamily\setminus\{C\} \colon \df(p,C)
< \df(p,C') \}.\] 
The Hausdorff Voronoi region of a point $c \in C$ is defined as:
\[ \hreg_\AbsFamily(c)  = \hreg_\AbsFamily(C) \cap \freg_C(c).\] 

The partitioning of the plane into 
Hausdorff Voronoi regions is called the \emph{Hausdorff Voronoi diagram}  of
$\AbsFamily$,  for brevity $\HVD(\AbsFamily)$, or simply $\HVD$.
We consider a 
refined version of  $\HVD(\AbsFamily)$,  where each
region $\hreg_{\AbsFamily}(C)$ is further subdivided into the finner
regions $\hreg_{\AbsFamily}(c)$, for $c\in C$. 
Fig.~\ref{fig:hvd} illustrates the Hausdorff Voronoi diagram of
a family of four clusters.
The convex hulls of the input clusters are illustrated in grey lines.
Solid lines indicate the Hausdorff Voronoi edges bounding the Voronoi
regions of individual clusters.
The dashed lines indicate the finer subdivision of a Hausdorff Voronoi region
 $\hreg_{\AbsFamily}(C)$ into  $\hreg_{\AbsFamily}(c)$, $c\in C$.

By the definition of a Hausdorff Voronoi region, $\hreg_{\AbsFamily}(c)=\emptyset$ for any point $c \in C$ that is not a vertex of $\conv C$.
Since such points are not relevant to the Hausdorff diagram, 
we assume 
that all
points in $C$ are also vertices of $\conv C$.

The structure of a Hausdorff Voronoi region is illustrated in  Fig.~\ref{fig:hreg}.
For a point $c\in C$, the boundary of $\hreg_{\AbsFamily}(c)$
consists 
 of two chains: (1)~the 
\emph{farthest boundary}, which is internal to
  $\hreg_{\AbsFamily}(C)$ (i.e., $\bd\hreg_{\AbsFamily}(c) \cap \bd\freg(c)
\subseteq \fskel(C)$); and (2)~the 
\emph{Hausdorff boundary} (i.e., $\bd\hreg_{\AbsFamily}(c) \cap \bd\hreg_{\AbsFamily}(C)$).
Neither chain can be empty, if $\cardinality{C}>1$,~\cite{EP-HVD-d_and_q}.

There are three types of vertices in a Hausdorff Voronoi diagram
(see Fig.~\ref{fig:hreg})~\cite{ep2004algorithmica}:
  (1)~\emph{Pure} Voronoi vertices, 
  which are equidistant to three clusters and appear on Hausdorff
  boundaries;  (2)~\emph{Mixed} Voronoi vertices, which are
  equidistant to three points of two clusters; and 
  (3)~\emph{Farthest}   Voronoi vertices of $\fskel(C)$, which appear
  only on   farthest boundaries.
Mixed Voronoi vertices are the points incident to a Hausdorff and a
farthest boundary.
Mixed vertices that are induced by two points in cluster $C$ and one point of another cluster are
  called $C$-\emph{mixed} vertices.
The farthest boundary of $\hreg_{\AbsFamily}(c)$ meets 
its Hausdorff boundary at 
a $C$-\emph{mixed} vertex, 
or  it extends to infinity, if $\hreg_{\AbsFamily}(c)$ is  unbounded.

Hausdorff Voronoi edges are polygonal lines  
 that connect pure Voronoi vertices and 
 separate the Voronoi regions of different
 clusters (see the solid lines in Fig.~\ref{fig:hvd}). Mixed Voronoi
 vertices correspond to the  breakpoints of these polygonal lines.

  \begin{figure}
  \begin{minipage}{0.48\linewidth}
  \centering 
  \includegraphics{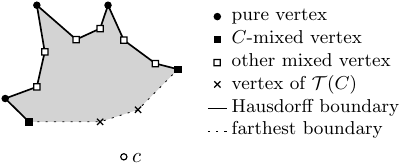}
  \caption{ Features of $\hreg_\AbsFamily(c)$} 
  \label{fig:hreg}
  \end{minipage}
  \hfill
  \begin{minipage}{0.51\linewidth}
  \centering
  \includegraphics{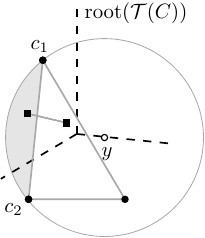}
  \caption{$\cdisk_y$  is partitioned by $\overline{c_1c_2}$ in $\cdisk_y^r$ (shaded)
    and $\cdisk_y^f$.
  The 2-point cluster $P$, illustrated in squares, is
  rear limiting w.r.t. the 3-point cluster $C$. 
}
  \label{fig:hreg&limiting}
\end{minipage}
\vspace*{-1ex}
\end{figure}

A line segment connecting two points in $C$ is called a \emph{chord}
of $C$. 
The closure of a Voronoi region is denoted as $\overline{reg}(\cdot)$.
In the following,  
we list some useful properties of the Hausdorff
Voronoi diagram.
\newpage
\begin{definition}[Rear{/}forward limiting
  cluster~\cite{EP-HVD-d_and_q}, see Fig.~\ref{fig:hreg&limiting}]  
    \label{def:limiting}

\noindent    \begin{itemize}
\item Let $y$ be a point
on an edge $e$ of $\fskel(C)$ induced by $c_1,c_2\in C$, i.e., $d_f(y,C)=d(y,c_1)=d(y,c_2)$.
Point $y$ partitions $\fskel(C)$ into two parts: $\fskel_y^r$ and
$\fskel_y^f$, where $\fskel_y^r$ is the subtree rooted at $y$  
as we traverse $\fskel(C)$  starting at its  root;
$\fskel_y^f$ is the complement of $\fskel_y^r$.

\item Let  $\cdisk_y$  be the disk centered at
$y$ of radius  $d_f(y,C)$. 
Chord $\overline{c_1c_2}$ partitions $\cdisk_y$ in two parts:
$\cdisk_y^r$ and $\cdisk_y^f$, where $\cdisk_y^r$
(shown shaded in Fig.~\ref{fig:hreg&limiting}) is the portion that contains the
points of $C$  that induce $\fskel_y^r$.    

 \item   A cluster $P$ enclosed in $\cdisk_y^r \cup \conv{C}$ (resp., in $\cdisk_y^f \cup \conv{C}$)
is called \emph{rear} (resp., \emph{forward}) limiting for $C$  with
respect to $y$.
\end{itemize}
  \end{definition}
 
\begin{property} [\cite{EP-HVD-d_and_q}]
\label{prop:limiting} 
If cluster $C$ has  a rear  (resp., forward) limiting cluster  $P$
with respect to  $y$ in $\fskel(C)$, 
then the entire $\fskel_y^r$ (resp., $\fskel_y^f$) is closer to
$P$ than to $C$.
\end{property}

Property~\ref{prop:limiting} implies 
Properties 2 and 3.

\begin{property}
Let $C,P \in \AbsFamily$. 
\begin{itemize}
\item[(a)]
\label{prop:connectivity}
 Region $\hreg_{\AbsFamily}(C)$ contains exactly one connected component of
 $\fskel(C)$, unless \mbox{$\hreg_{\AbsFamily}(C) = \emptyset$}.

\item[(b)] \label{prop:vertex}
Let $v$ be a vertex in  $\fskel(C)$. 
If $d_f(v,P) < d_f(v,C)$, then  only one of the subtrees incident to
$v$ may intersect $\hreg(C)$. 
\item[(c)] \label{prop:bisector_wrt_edge}
Let $e=uv$ be an edge in $\fskel(C)$.
If both $u$ and $v$ are closer to $P$ than to $C$, then  
$\hreg_{\AbsFamily}(C)\cap e=\emptyset$.
\end{itemize}
\end{property}

\begin{property} [\cite{EP-HVD-d_and_q}]
\label{prop:killingpair} 
Region ${\hreg_{\AbsFamily}(C)} = \emptyset$ if
and only if
one of the following conditions holds: (1) there is a cluster in $\AbsFamily$ entirely contained in $\conv{C}$; (2)~there is a pair of clusters 
in $\AbsFamily$ such that one is forward limiting and the other is rear limiting with respect to the same point $y \in \fskel(C)$.
\end{property}

The cluster or pair of clusters of Property~\ref{prop:killingpair} is
called a  \emph{killer} or a  \emph{killing pair} of $C$ respectively.

\begin{property} [\cite{EP-HVD-d_and_q}]
\label{prop:star-shaped}
For any point $x \in {\hreg_{\AbsFamily}(c)}$
the line segment $cx \cap \freg_C(c)$ lies entirely in $\hreg_{\AbsFamily}(c)$.
\end{property}

\section{A Randomized Incremental Algorithm}
\label{sec:algorithm}
Let $C_1, \dots, C_{\nclus}$ be a random permutation of
the input family $\InpSet$ of clusters. 
Let $\InpSet_i = \{C_1, \dots, C_i\}$, $1 \leq i \leq {\nclus}$, be the  set of the first $i$
clusters in 
this permutation. 
Our algorithm incrementally computes $\HVD(F_i)$, for \mbox{$1 < i \leq \nclus$}, starting
with $\HVD(F_1) = \FVD(C_1)$.
At step~$i$, we insert cluster $C_i$ in $\HVD(\InpSet_{i-1})$ and derive  $\HVD(\InpSet_{i})$.
To this goal, we identify a  \emph{representative} point
$\rep \in \hreg_{F_i}(C_i)$ or we determine that no such point exists. 
If $\rep$ exists,
we trace the boundary of  $\hreg_{F_{i}}(C_{i})$, and update the diagram to
$\HVD(F_{i})$. Else, 
we conclude that $\hreg_{F_i}(C_i)=\emptyset$  and $\HVD(F_i) =  \HVD(F_{i-1})$. 

The main challenge of our algorithm is to efficiently identify a representative point $\rep$ or conclude that no such point exists. 
Then, the tracing of $\hreg_{F_i}(C_i)$ can be  done similarly to~\cite{EP-HVD-d_and_q}, 
in time proportional to the complexity of the new region, plus the
complexity of the deleted portion of the diagram, times $O(\log{n})$
(see Section~\ref{sec:tracing}). 
In the remaining of this section we focus on identifying a
representative point $\rep$.
We skip the subscript ``$F_i$'' and let $\hreg(C_i)$ stand for $\hreg_{F_{i}}(C_i)$.

Property~\ref{prop:connectivity}a implies three possibilities for $\hreg(C_i)$: 
(1) $\hreg(C_i) \cap \fskel(C_i)$ contains a vertex of $\fskel(C_i)$;
(2) $\hreg(C_i)$ intersects exactly one edge of $\fskel(C_i)$;
and (3) $\hreg(C_i)$ is empty.
The edge of case (2) is called a \emph{candidate edge} (see Definition~\ref{def:candidate}). 
In Fig.~\ref{fig:hvd}, the bounded region in the middle illustrates
case (2), while the three unbounded regions illustrate case (1).

To identify case
(1), it is enough to 
perform point location  of the vertices in $\fskel(C_i)$ in
$\HVD(F_{i-1})$. If any vertex $v$ is found
to be closer to $C_i$ than to its owner in $\HVD(F_{i-1})$, then
$v$ can serve as 
a representative point, i.e., $\rep = v$.
Suppose that no vertex of  $\fskel(C_i)$ satisfies case (1).
Then we look for a \emph{candidate edge} that may satisfy case (2).

\begin{definition}[candidate edge]
  \label{def:candidate} 
Let $uv$ be an edge of $\fskel(C_i)$ and let 
 $Q^u, Q^v $ be the clusters in $F_{i-1}$ closest to $u$
and $v$ respectively. 
Edge $uv$ is called a \emph{candidate edge} if $Q^{u} \neq Q^{v}$ and $uv$ satisfies
the following predicate:
\(
  \cand(uv) \allowbreak = \allowbreak
    \left({\df(u  ,Q^{u  }) < \df(u  , C_{i}) < \df(u  ,Q^{v  })}\right)
     \allowbreak \land \allowbreak
    \left({\df(v  ,Q^{v  }) < \df(v  , C_{i}) < \df(v  ,Q^{u  })}\right)
\).
\end{definition}

\begin{lemma}
\label{lemma:cand}
If there is an edge $uv$ of $\fskel(C_i)$ that is a candidate edge, then either  $\hreg(C_i)$ intersects $uv$ or $\hreg(C_i) = \emptyset$.
\end{lemma}

\begin{proof}
Consider the clusters $Q^u$ and $Q^v$ of
Definition~\ref{def:candidate}. By Property~\ref{prop:vertex}b,
only one of the subtrees of $u$ can be closer to $C_i$ than to $Q^u$.
Because $uv$ is a candidate edge satisfying Definition~\ref{def:candidate},
 $v$ lies in this subtree, 
that is, only the subtree of $u$ that contains $uv$ may intersect
$\hreg(C_i)$. Symmetrically, only the subtree of $v$ that contains
$uv$ may  intersect $\hreg(C_i)$. The intersection of these two
subtrees is exactly the edge $uv$. 
Thus, $\fskel(C_i) \cap \hreg(C_i) \subset uv$, or $\fskel(C_i) \cap \hreg(C_i) = \emptyset$. In the latter case, by Property~\ref{prop:connectivity}(a), $\hreg(C_i) = \emptyset$.
\end{proof}

Lemma~\ref{lemma:cand} implies that, given a candidate edge $uv$,
it suffices to search on $uv$ to identify a representative point.  
Furthermore, if  such a point cannot be found on $uv$, then
$\hreg(C_i) = \emptyset$. 
We can 
search for a representative point as follows:
Traverse $\fskel(C_i)$, starting at its root,
checking vertices and possibly pruning appropriate subtrees
according to Properties~\ref{prop:vertex}b and~\ref{prop:vertex}c. 
During the traversal, either determine $\rep$ as a vertex of $\fskel(C_{i})$, or
determine a candidate edge $uv$, or conclude that
$\hreg(C_{i})=\emptyset$.

When  a candidate edge $uv$ is determined, we still need to identify a 
representative point $\rep$ on $uv$ or determine that $\hreg(C_i) =\emptyset$.
This is achieved by performing a 
\emph{parametric point location query} in  $\HVD(F_{i-1})$ for edge $uv$
 as given in the following definition.

\begin{definition}[Parametric point location query]
      \label{def:param-pl}
       Given a family of clusters $\AbsFamily$, $\HVD(\AbsFamily)$, a cluster $C \not\in \AbsFamily$,
       and a candidate edge $uv \subset \fskel(C)$,  determine
        a cluster $P \in \AbsFamily$ and
        a point $\rep \in uv$ 
         such that $\rep \in \hreg_{\AbsFamily}(P)$ and $\df(\rep, P) = \df(\rep, C)$.         
If such a point does not exist, return \emph{nil}.
\end{definition}

The performance of the  parametric point location query 
induces 
the time complexity of our algorithm.
To perform parametric as well as ordinary point location 
efficiently, we store 
$\HVD(F)$ in a hierarchical data structure, called the Voronoi
hierarchy. This data structure is described in  Section~\ref{sec:vh}
and the parametric point location query is detailed in Section~\ref{sec:vh-parPL}.
The parametric point location requires to also answer 
an additional non-standard  location query on the static farthest Voronoi
diagram of a given cluster,  called a \emph{segment query}. 
A data structure to efficiently answer segment queries 
 is given in the following section. 

\section{Centroid Decomposition}
\label{sec:sep_decomp}
This section describes a data structure, called the 
\emph{centroid
decomposition}, that can efficiently answer queries related to
point location  on a
planar subdivision induced by a tree structure. 
The centroid decomposition was introduced by Megiddo et al.~\cite{megiddo-centroid-decomp}, 
and we use it 
 in this paper to efficiently
perform  \emph{segment queries} 
on the farthest Voronoi diagram of a cluster. 
Segment queries are used during parametric point location in the Hausdorff Voronoi diagram. 
The query is defined as follows.  

\label{sec:sd-segm}
\begin{definition}[Segment query]
  \label{def:sd-seg} 
Consider two clusters $C,P$.
Given a segment $uv \subset \fskel(C)$ 
such that $\df(u,C) < \df(u,P)$ and $\df(v,C) > \df(v,P)$,
find the point $x \in uv$ that is equidistant from $C$ and $P$, i.e., $\df(x, C) = \df(x,P)$. (See Fig.~\ref{fig:sd-test}b.)
\end{definition}

\begin{figure}
\begin{minipage}{0.49\linewidth}
\centering
\includegraphics{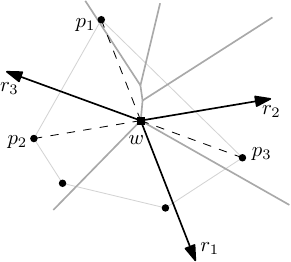} \\
(a)
\end{minipage}
\begin{minipage}{0.49\linewidth}
\centering
\includegraphics{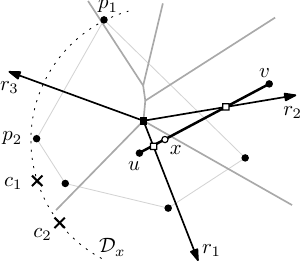} \\
(b)
\end{minipage}

\caption{A cluster $P$ (disks), $\FVD(P)$ (thick grey lines) and (a) the rays $r_i$ corresponding to vertex $w$; (b) the segment  query for $uv$ outputs point  $x$}
\label{fig:sd-test}
\end{figure}

Let $P$ be a cluster of points and let $\Sep(P)$ denote the 
centroid
decomposition for $\FVD(P)$.
We first define $\Sep(P)$, and then we describe how to use it to
perform  an ordinary  point location query on $\FVD(P)$.
Finally, we adapt the point location query to perform a segment query
on $\FVD(P)$.
\deleted{
We start with defining the 
centroid
decomposition for $\FVD(P)$, denoted as $\Sep(P)$.
We then describe an ordinary  point location query in $\FVD(P)$ using
$\Sep(P)$,  
and finally we adapt the point location query to perform a segment query
on $\FVD(P)$.
}

\paragraph{Definition of $\Sep(P)$.} Any tree
with $h$ vertices has a vertex $v$, called the 
\emph{centroid},  whose removal decomposes the tree into
subtrees of at most  $h/2$ vertices each~\cite{megiddo-centroid-decomp}. 
Exploiting this fact, we 
build $\Sep(P)$ as a balanced  tree
whose nodes correspond to Voronoi vertices of $\FVD(P)$ as follows:

\begin{itemize}
  \item
Find the centroid $w$ of  $\fskel(P)$. 
Create a node for $w$ and assign it as the root of $\Sep(P)$. 
\item 
Remove $w$ from $\fskel(P)$.  Recursively
build the 
centroid decomposition 
trees for  
the connected components of $\fskel(P)$, which are  incident to $w$,
and  
link them as subtrees of $w$. 
\end{itemize}

\paragraph{Point location on $\FVD(P)$ using $\Sep(P)$.}
Given a query point $q$, we perform {point location} 
as follows. Starting from the root,
we traverse  $\Sep(P)$, testing $q$  against the current node of
$\Sep(P)$, 
as explained in the next paragraph.
Every time 
we choose one of  
the node's subtrees to continue, until a leaf node is reached. 
Among the three points in $P$ inducing the Voronoi vertex of the leaf node,
we choose the one farthest from $q$.
Then, 
point $q$ belongs in the farthest Voronoi region of the chosen point.

Testing $q$ against a node $\alpha$ of $\Sep(P)$  
is performed  following
Aronov~{et al.}~\cite{aronov06data}.
In particular, let $w$ be  the Voronoi vertex  of $\FVD(P)$ 
corresponding to 
node $\alpha$. 
Let $p_1, p_2, p_3\in P$ 
have farthest Voronoi regions incident to $w$.
Consider the rays $r_i, i =
1,2,3,$ originating at $w$ and having direction
$\overrightarrow{p_iw}$ respectively (see Fig.~\ref{fig:sd-test}a).
The rays $r_i$ subdivide the plane into three sectors.
Among the subtrees of $\Sep(P)$ incident to $\alpha$,
pick the one that corresponds to the sector containing $q$.
Correctness 
is implied by 
the construction of $\Sep(P)$ and  the following lemma.

\begin{lemma}
\label{lemma:sectors}
Rays $r_i, i=1,2,3,$ subdivide the plane
into three sectors, where each sector contains
 exactly one connected component of
$\fskel(P)\setminus\{w\}$. 
\end{lemma}
\begin{proof}
It is well known
that for any point  $t \in \freg_P(p)$, $p\in P$, 
the ray originating at $t$
and having direction  $\overrightarrow{pt}$  
is entirely contained in
$\freg_P(p)$. Thus, no ray $r_i$ can intersect the edges of $\fskel(P)$. 
Since the rays $r_i$ lie in three distinct regions of $\FVD(P)$, there
is a component of 
$\fskel(P)\setminus\{w\}$ in each of the three sectors formed by 
these three rays. The claim follows.
\end{proof}

\paragraph{Segment query on $\FVD(P)$ using $\Sep(P)$.}
The segment query can be performed similarly 
to a point location query within the same time complexity.   
The difference is in the testing of 
segment $uv$ against a node $\alpha$ of $\Sep(P)$.
Let $w$ be the Voronoi vertex of $\FVD(P)$ 
corresponding  to $\alpha$.
Let rays  $r_i, i=1,2,3$, emanate from $w$ as defined above 
(see Fig.~\ref{fig:sd-test}b).

Consider the (at most two) intersection points of the rays with $uv$.
If any of these points 
is equidistant to $C$ and $P$, return it.
Otherwise, these intersection points break $uv$ into (at most  three) subsegments, each lying in one of the three sectors formed by the rays $r_i$. 
Among these subsegments, pick subsegment $u'v'$ such that
$\df(u',C) < \df(u',P)$ and $\df(v',C) > \df(v',P)$, and continue with the 
child of $\alpha$ in $\Sep(P)$ whose Voronoi vertex lies in the same  sector as $u'v'$.

If $\alpha$ is a leaf of $\Sep(P)$, 
let $e$ be the edge of $\fskel(P)$ incident to the vertex $w$ that
lies in the same  sector as $u'v'$.
Let $p_1, p_2$  be the points in $P$ that induce the edge $e$, and 
let $c_1, c_2$ be the points in $C$ that induce $uv$.
Since $\df(v,P) < \df(v,C) = \df(v,c_1) = \df(v,c_2)$,
the closed disk centered at $v$ and passing through $c_1, c_2$
must contain both $p_1$ and $p_2$.
Since $C$ and $P$ are non-crossing, both $p_1$ and $p_2$ lie to the same side of the chord $\overline{c_1c_2}$.
Thus, one of the two closed disks defined by points $p_1, c_1, c_2$ or by points  $p_1, c_1, c_2$
must contain both $p_1$ and $p_2$.
Return as an answer to the segment query the center of this disk.
In Fig.~\ref{fig:sd-test}b, such a disk $\cdisk_x$ has
a dotted arc
on its boundary and its center $x$ is shown as an unfilled circle.

\begin{lemma}
\label{lemma:sepdecomp}
  The 
  centroid decomposition $\Sep(P)$ of a cluster $P$
  can be built in $O(n_p\log{n_p})$ time, where $n_p$ is the number of
  vertices of $\FVD(P)$.
  Both the point location and the segment query in $\Sep(P)$
require $O(\log{n_p})$ time.
\end{lemma}

\begin{proof}
Given a subtree of $\fskel(P)$, its centroid can be
computed in $O(h)$ time~\cite{megiddo-centroid-decomp}, where $h$ is
the number of vertices in this subtree. Building $\Sep(P) $ requires
to recursively compute the centroids of its subtrees, each of size at
most  half the size of $P$. This implies an $O(n_p\log{n_p})$ total time to build $\Sep(P)$.

The point location query consists of  $O(\log{n_p})$ tests of a query point against a node of $\Sep(P)$. 
Each test involves a constant number of points and rays, thus, it can be performed in constant time. 
The same argument works for the test of a segment against a node of
$\Sep(P)$ during a segment query, which implies the same
$O(\log{n_p})$ time bound. 
\end{proof}

\section{The Voronoi Hierarchy for the Hausdorff Voronoi Diagram}
\label{sec:vh}
We describe a randomized 
semi-dynamic data structure to store the Hausdorff Voronoi
diagram,  
which supports
insertion of a cluster, and
point location 
queries (both ordinary and parametric). 
It augments
the Voronoi  hierarchy~\cite{Boissonat-curved_vd,KARAVELAS-vd_of_co_in_the_plane} with the
ability to handle the generalized Voronoi features present in the
Hausdorff diagram, 
such as sites of non-constant complexity, sites that are not
entirely contained in their regions,  and empty Voronoi
regions. 
The Voronoi hierarchy is inspired by the  Delaunay
hierarchy~\cite{devillers-delaunay_hier} 
that yields an optimal
randomized incremental construction of the Delaunay triangulation.
The  Delaunay hierarchy 
can be  considered 
a 2-dimensional version of the skip lists of Pugh~\cite{pugh}, 
\cite{Boissonat-curved_vd}.
We refer to our adaptation as the 
\emph{Hausdorff Voronoi hierarchy}.

\begin{definition}[\cite{KARAVELAS-vd_of_co_in_the_plane,devillers-delaunay_hier}]
The \emph{Voronoi hierarchy} of a set of sites $S$ is a sequence
of Voronoi diagrams, 
 $\VD(S^{(\ell)}), \ell = 0,\ldots,h$, where the sets $S^{(\ell)}$ form a hierarchy of subsets 
of $S$ 
built  as follows. 
$S^{(0)} = S$, and for $\ell = 1, \ldots, {\nlevs}$, 
$\SiteSet^{(\ell)}$ is a  random sample of 
$\SiteSet^{(\ell-1)}$ following a Bernoulli distribution with a fixed constant
 parameter $\beta \in (0,1)$. 
 We refer to $\ell$ as ``level of the Voronoi hierarchy''. 
\end{definition}

To perform \emph{point location} for a query
point $q$ 
in the Voronoi hierarchy, we start at the last level $h$, and for each level $\ell$,
we determine the site $s^\ell \in S^{(\ell)}$ that is closest to $q$
 by performing a \emph{walk}.
Each step of the walk moves from the current site to one of its neighbors 
such that the distance to $q$ is reduced.
To determine an appropriate neighbor,  binary search may be
used~\cite{KARAVELAS-vd_of_co_in_the_plane}.
A walk at level $\ell$ starts at  $\site^{\ell+1}$.
The answer to the query is $\site^0$.

\paragraph{}
For the Hausdorff Voronoi diagram, a first difference to consider is
that clusters are not of constant complexity and 
that $n$ can be $w(k)$.
Recall that $k$ is number of sites (clusters) and $n$ is their total complexity $(k\leq n)$.

Nevertheless, the following lemma shows that the complexity of the Hausdorff Voronoi hierarchy 
remains 
expected-$O(n)$ as for the original
hierarchy~\cite{devillers-delaunay_hier,KARAVELAS-vd_of_co_in_the_plane}.

\begin{lemma}
   \label{lemma:vh-size}  
Consider the Voronoi diagram of a family of $k$ sites of total complexity
$n$, where the size of the diagram is also $O(n)$.
Then the Voronoi hierarchy 
for such diagram
has  expected size $O(n)$ and 
    expected number of levels $O(\log{{\nclus}})$.
 \end{lemma}
 
 \begin{proof}
  Let $\complexity{\SiteSet^{(\ell)}}$ 
 denote the total complexity of the sites at a level $\SiteSet^{(\ell)}$.
  For any site $s \in \SiteSet$, the probability that  $s$ appears in 
  $\SiteSet^{(\ell)}$ is $\beta^\ell$.
  Then, 
 the expectation of $\complexity{S^{(\ell)}}$ is 
  $\E[\complexity{\SiteSet^{(\ell)}}] = \beta^\ell \complexity{\SiteSet}
  = \beta^\ell n$,
  and the expected size of the Voronoi diagram at
  level $\ell$ is $O(\beta^\ell n)$.
  The expected size of the hierarchy is
  \[ \sum_{\ell = 0}^{\infty} O(\E[\complexity{\SiteSet^{(\ell)}}]) =
  \sum_{\ell = 0}^{\infty} O(\beta^\ell n) =
\frac{1}{1-\beta}O(n) = O(n).\]
The bound on the expected number of levels follows immediately
from properties of the Bernoulli distribution
\cite{devillers-delaunay_hier,KARAVELAS-vd_of_co_in_the_plane}.
\end{proof}

To adapt the Voronoi hierarchy for the Hausdorff Voronoi diagram, several difficulties have to be addressed. 
When performing a walk at a level $\ell$ of the hierarchy, at 
each step we need to reduce  the distance 
between the current cluster $C$ and the query point $q$.  However, the farthest  distance
$\df(q,C)$ may be realized by a point 
$c\in C$ that has 
an empty Voronoi region at level $\ell$.
Thus, instead of $\df(\cdot,\cdot)$, we base the walk on a slightly
different distance function, which reflects the diagram better, 
 and
which equals    $\df(\cdot,\cdot)$ at the end of the walk. 

In addition, the neighbors of a Hausdorff Voronoi region do not have a natural
ordering, and thus, it is not easy to use binary search when performing one step in the walk.
To address these problems, we redesign 
one step of the walk in Section~\ref{subsec:one-step}. 
Then, point location can be  performed as in the ordinary
hierarchy. In Section~\ref{sec:vh-parPL}, 
 we describe the
parametric point location query
that is needed for our algorithm.
Empty Voronoi regions in the Hausdorff diagram pose another major difficulty when updating the
 hierarchy because they complicate the  transition
between the hierarchy levels.
In Section~\ref{subsec:update-vh}, we show how to maintain the hierarchy and deal with regions that
become empty.

In our modified hierarchy, 
a walk at level $\ell$ does not necessarily start from the same cluster where it stopped at level $\ell +1$, 
 but possibly from another cluster that is closer to $q$. The following lemma 
shows that the 
expected 
length of the walk on one level of the Hausdorff Voronoi
hierarchy is constant. It is a simple modification of
 \cite[Lemma~9]{KARAVELAS-vd_of_co_in_the_plane}.

\begin{lemma}
  \label{lemma:walkconstexp}%
  Let $\site_0^\ell, 
  \ldots, \site_r^\ell = \site^{\ell}$ 
  be the sequence of
  sites visited at level $\ell$ during  point
  location for  a query point $q$.
  Assuming that $\df(q, \site_i^\ell) < \df(q, \site_{i-1}^{\ell})$,
   \mbox{$i = 1,\dots,r$}, and either $\site^{\ell+1}=\site_0^\ell$, or 
  $\df(q, \site_0^\ell) < \df(q, \site^{\ell+1})$,
  the expectation of the length $r$ of the walk at level $\ell$ is
  constant. 
\end{lemma}
  
  \begin{proof} 
  Each site visited by the walk at level $\ell$
    is closer to $q$ than $\site^{\ell+1}$, and thus does not belong to level $\ell+1$.
 The probability that there are $t$ such sites is $\beta(1-\beta)^{t-1}$. 
Thus, the expected number of sites visited at level $\ell$ is at most
  \[     \sum\limits_{t=1}^{\cardinality{\SiteSet^{(\ell)}}}{t(1-\beta)^{t-1}\beta}
     <
      \beta\sum\limits_{t=1}^{\infty}{t(1-\beta)^{t-1}} = \frac{1}{\beta}. \] 
  \end{proof}

\subsection{One step of the walk}
\label{subsec:one-step}

Let $\ell \in \{0,\dots,h\}$  be a level in the Hausdorff Voronoi
hierarchy of $\AbsFamily$. Let
$\hregl(\cdot)$ denote $\hreg_{\AbsFamily^{(\ell)}}(\cdot)$ 
and let $\hreglc(\cdot)$ denote the closure of this region.

A point $c\in C$ is called \emph{active} if
$\hregl(c) \neq \emptyset$.
Let $\hat{\curClus}$ denote the set of all active points in a cluster $C$.
The walk to locate a query point $q$ uses  
the farthest distance to the active points of a cluster $C$ 
as opposed to the farthest
distance to all points 
of $C$.
One step of the walk is defined as follows.

\begin{definition}[a step of the walk at level $\ell$]
\label{def:onestep}
Given a query point $q$ and a cluster $\curClus \in
\AbsFamily^{(\ell)}$ such that $q\not\in   
\hreglc(\curClus)$,
determine $Q \in \AbsFamily^{(\ell)}$ such that 
 $\hregl(Q)$ is adjacent to $\hregl(C)$  and $\df(q, \hat{Q}) < \df(q, \hat{\curClus})$.
If $q \in  \hreglc(\curClus)$ then $Q=C$.
\end{definition}
 
\begin{figure}
\begin{minipage}{0.49\linewidth}
  \centering
  {\includegraphics  {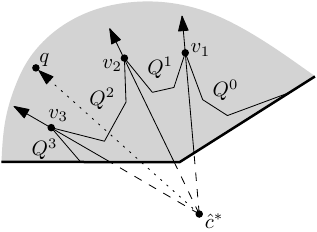}}
  \caption{One step of a walk  for a query point $q$ and a starting cluster $C$}  
\label{fig:regionpart}
\end{minipage}
\hfill
\begin{minipage}{0.49\linewidth}
\centering
\includegraphics {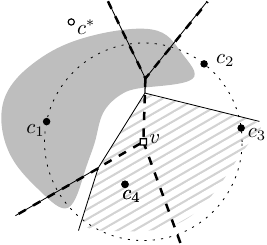}
\caption{Illustration of the proof of Lemma~\ref{lemma:tangents}}
\label{fig:tangents}
\end{minipage}
\end{figure}

To perform one step 
of the walk we use the set of active points $\hat{\curClus}$.
We store  $\hat{\curClus}$ as a circular list  of its points 
in the order of its convex hull. Each point $c \in \hat{\curClus}$
has a link to  the ordered list  of pure
Voronoi vertices  $v_1,\ldots,v_j$  on the boundary of $\hregl(c)$. 
(Recall from Section~\ref{sec:prelim} that pure Voronoi vertices are equidistant to three different
clusters.) Let $Q^0,\ldots,Q^{j+1}$ be  the corresponding list of clusters whose 
Voronoi regions are adjacent to $\hregl(c)$ (see Fig.~\ref{fig:regionpart}).
We determine cluster $Q$ by binary search in these lists.
The detailed algorithm is given in Procedure~\ref{alg:step}.

For the rest of this section, let $c^*$ (resp.,  $\hat{c}^*$) be the 
 point in $C$ (resp., in $\hat{C}$) that is farthest from $q$ ($\df(q,c^*) = \max_{c \in C}{\df(q,c)}$ 
 and $\df(q,\hat{c}^*) = \max_{c \in \hat{C}}{\df(q,c)}$).
Let $c_1,c_2 \in \hat{C}$ be 
the active points immediately following and preceding $c^*$ respectively
on the boundary of $\conv{C}$.
In the following, we use $c_1$ and $c_2$ to determine $\hat{c}^*$.

\begin{algorithm}[H]
  \caption{A step 
  of the walk at level $\ell$}
  \label{alg:step}
  \begin{algorithmic}[1]
 \State Determine ${c}^*$ by locating $q$ in $\FVD(C)$.
 \If {$c^* \in \hat{C}$}
 \State{ Let $\hat{c}^* = c^*$.} 
 \Else \State{Let 
$\hat{c}^*$ be the point in $\{c_1, c_2\}$ that is the farthest from $q$.} (See Lemma~\ref{lemma:tangents})
\EndIf
\State Let $Q = Q^i$
 such that ray $\overrightarrow{\hat{c}^*q}$ follows 
$\overrightarrow{\hat{c}^*v_i}$ and/or precedes
$\overrightarrow{\hat{c}^*v_{i+1}}$. (See Fig.~\ref{fig:regionpart})
\If {$\df(q, \hat{Q}) < \df(q, \hat{\curClus})$}
\State \Return $Q$.
\Else 
 \State \Return $C$.  
 \EndIf
  \end{algorithmic}
\end{algorithm} 
 
To establish the correctness of Procedure~\ref{alg:step} we need to
prove correctness for Lines~5 and 6--9.
The following lemma shows that $\hat{c}^*\in\{c_1,c_2\}$ (if $c^*$ is
not active), and thus, it establishes the 
correctness of Line~5.

\begin{lemma}
\label{lemma:tangents}
Assuming  ${c}^* \not\in \hat{C}$,
$\freg_{\curClus}(c^*) \subset \freg_{\hat{\curClus}}(c_1)
  \cup \freg_{\hat{\curClus}}(c_2)$.
\end{lemma}

\begin{proof}
Let $C' = \hat{C}\cup\{c^*\}$.  Since $C' \subseteq C$,
$\freg_{C}(c^*) \subseteq \freg_{C'}(c^*)$. Thus 
it is enough to prove that
$\freg_{C'}(c^*) \subset \freg_{\hat{C}}(c_1) \cup \freg_{\hat{C}}(c_2)$.

Suppose for the sake of contradiction that there is $c_3 \in
  \hat{C} \setminus \{c_1,c_2\}$ such that 
 $\freg_{{C'}}(c^*) \cap \freg_{\hat{\curClus}}(c_3) \neq \emptyset$  (see Fig.~\ref{fig:tangents}).
Then $\freg_{{\curClus'}}(c^*)$  (shown striped in Fig.~\ref{fig:tangents}) 
has at least three neighbors in $\FVD({C'})$, which implies that
 $\freg_{C'}(c^*)$ contains at least one vertex $v$ of $\fskel(\hat{C})$.
Fig.~\ref{fig:tangents} shows  $\hregl(C)$ as a grey area,
$\fskel(\hat{C})$ in  thick dashed lines, 
 and  $\fskel({C'})$ in thin  solid lines. 
Since all points in $\hat{C}$ have non-empty Voronoi regions in
$\HVD_F^{(\ell)}(C)$, all vertices of $\fskel(\hat{C})$ must be contained in 
 $\hreg^{(\ell)}_{F}(C)$.
Thus, $v \in \hreg^{(\ell)}_{F}(C)$, which implies that $c^*$ is
active; a contradiction.
\end{proof}

\begin{figure}
\begin{minipage}{0.49\linewidth}
\centering
\includegraphics{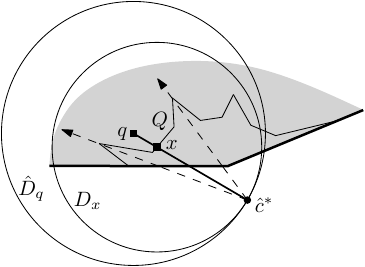}
\\
(a)
\end{minipage}
\begin{minipage}{0.49\linewidth}
\centering
\includegraphics{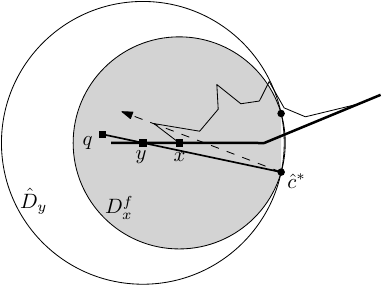}
\\
(b)
\end{minipage}
\caption{Two cases of the proof of Lemma~\ref{lemma:d-reduced}: (a) $\hat{c}^*q$ intersects
the Hausdorff boundary of $\hat{c}^*$; (b) $\hat{c}^*q$ does not intersect it}
\label{fig:d-reduced}
\end{figure}

\begin{lemma}
\label{lemma:d-reduced}
Let $Q $ be the cluster determined at Line~6 of Procedure~1.
Then $\df(q, \hat{Q}) < \df(q, \hat{\curClus})$ if and only if 
$q\not \in \hreglc(\curClus)$.
\end{lemma}

\begin{proof}
Suppose $q \in\hreglc(\curClus)$.
Let $\hat{F} = \{\hat{P}, P \in F^{(\ell)}\}$ be the family of sets of active points of all
clusters in $F^{(\ell)}$.  
Clearly, $\HVD(\hat{F})$ is identical to $\HVD(F^{(\ell)})$. 
Since  $\hreg_{F}^{(\ell)}(C) = \hreg_{\hat{F}}(\hat{C})$ it follows
that $q \in \overline{\hreg}_{\hat{F}}(\hat{C})$. Therefore, 
 $\df(q,\hat{C}) \leq
\df(q,\hat{Q})$.

Suppose $q \not\in \hreglc(C)$. Let $\hat{D}_q$ 
be the closed disk centered at $q$ with radius $|\hat{c}^*q|$. 
We prove    that $\hat{Q}$ is enclosed in $\hat{D}_q$, which is
equivalent to $\df(q, \hat{Q}) < \df(q, \hat{\curClus})$.
There are two cases: 

\begin{enumerate}
      \item
Suppose that
segment 
$\hat{c}^*q$ intersects the Hausdorff boundary of
$\hregl(\hat{c}^*)$ at a point $x$.
(By Property~\ref{prop:star-shaped}, ${\hat{c}^*q}$ may intersect this
boundary at most once). Let $\cdisk_x$ be the closed disk centered at $x$ with radius $\df(x,C)$.
Since $x$ lies on the boundary between $\hregl(C)$ and $\hregl(Q)$,  $\df(x, C) = \df(x,Q)$.
        Equivalently,  $Q \subset \cdisk_x$. 
Since $q \not\in \hregl(\hat{c}^*)$, then $\df(q,\hat{Q}) < \df(q,\hat{C}) = |q\hat{c}^*|$. 
Thus $\hat{Q}$ is enclosed in $\hat{D}_q$. See Fig.~\ref{fig:d-reduced}a.

\item 
 Suppose that $\hat{c}^*q$ does not intersect
  the Hausdorff boundary of $\hregl(\hat{c}^*)$. 
  Since  $q \not\in \hregl(C)$,  the situation is as in  Fig.~\ref{fig:d-reduced}b.
  Let $y$ be a point where 
  ${\hat{c}^*q}$ intersects $\bd \freg_{ \hat{\curClus}}(\hat{c}^*)$; let $e$ 
  be the edge of $\bd \freg_{ \hat{\curClus}}(\hat{c}^*)$ that contains $y$, and let $x$ be 
  the $C$-mixed Voronoi vertex encountered first as we traverse $\bd \freg_{ \hat{\curClus}}(\hat{c}^*)$ from $y$ towards $\hregl(C)$.
 Note that $Q$ is the cluster whose Voronoi region is incident to $x$.
  Cluster 
  $\hat{Q}$ is either rear or forward limiting with respect to  $\hat{\curClus}$ (see Definition~\ref{def:limiting}).
   Without loss of generality, let $\hat{Q}$ be forward limiting, that is, 
  $\hat{Q} \subset \cldf_x \cup \conv \hat{\curClus}$. 
Since $y \not\in \hregl(C)$ and $y \in \overline{\freg_{\hat{C}}}(\hat{c}^*)$, 
then  $\cldf_x \cup \conv \hat{\curClus} \subset \hat{D}_y \subset \hat{D}_q$, where $\hat{D}_y$
the closed disk centered at $y$ with radius $\df(y, \hat{C})$.  
  Thus, $\hat{Q}$ 
  is enclosed in $\hat{D}_q$. 
  \end{enumerate} 
\end{proof}

\begin{lemma}
\label{lemma:step}
One step of the  walk is performed in $O(\log{n})$ time.  
\end{lemma}
\begin{proof}
We analyze the time complexity of Procedure~1. 
 In Line~1, point ${c}^*$ is determined  by locating $q$ in $\FVD(C)$ in $O(\log{n})$ time. 
 In Line~5, points $c_1,c_2$ are determined in time $O(\log{n})$ by drawing
the tangents from $c^*$ to $\hat{C}$, which is a convex chain.
In Line~6, cluster $Q$ is found by binary
search in the list $v_1,\ldots,v_j$.
Thus, all steps are performed within time  $O(\log n)$. Correctness is
established by Lemmas \ref{lemma:tangents} and \ref{lemma:d-reduced}.
\end{proof}

\subsection{Parametric point location in the Voronoi hierarchy}
\label{sec:vh-parPL} 
In this section we show how to perform  parametric point location
on a candidate edge $uv$ of $\fskel(C)$.
 Recall from Section~\ref{sec:algorithm} the definition of a candidate edge
 (Definition~\ref{def:candidate}) and of a parametric point 
 location query (Definition~\ref{def:param-pl}).
 
We follow the same top-down traversal of the hierarchy, as for
the ordinary point location. Starting  at the last level ${\nlevs}$ of the hierarchy,
at each level $\ell$, we search for 
  a cluster $Q^{\ell} \in \AbsFamily^{(\ell)}$ and a point
  $u^{\ell} \in uv$  such that  $u^{\ell} \in \hregl(Q^{\ell})$ and
$\df(u^{\ell},\newc) = \df(u^{\ell}, Q^{\ell})$. 
The answer to the query is the cluster $C^0$ and the point $u^0$ of level $0$.
If at any level $\ell$ we find out that 
the desired cluster $Q^{\ell}$ or 
point $u^\ell$ do not exist, the answer to the parametric point location query is \emph{nil}.
At level $\ell$, we determine a sequence 
$(a_j)_{j=0}^r$ of points on the line segment $uv$, such that $a_0 = u^{\ell+1}$  and 
$a_r =  u^\ell$. Let $Q^{a_j}$,  $j = 1,\ldots, r$, 
denote the cluster in $F^{(\ell)}$  such that $a_j \in \hregl(Q^{a_j})$. 
For each $j = 0, \ldots, r$   point $a_{j+1}$ is 
derived from $a_j$ so that it is   
equidistant to $Q^{a_j}$ and $C$ ($\df(a_{j+1}, C) = \df(a_{j+1}, Q^{a_j})$). 

The algorithm to perform
parametric point location is given in Procedure~\ref{alg:par-pl}. 
At level $\ell$, we determine $a_{j+1}$ from $a_j$ by performing a
walk starting at $Q^{a_j}$, and by performing a segment query (see Definition~\ref{def:sd-seg})
for a subsegment of $uv$ on $\FVD(Q^{a_j})$.

\begin{algorithm} 
  \caption{Parametric point location  
on candidate edge $uv$} 
\label{alg:par-pl}
  \begin{algorithmic}[1]
\State  Find $Q^h$ and $u^h$ (by brute force).
\For{$\ell = (h-1)$ \textbf{downto} $0$}
\State Set $a_0 = u^{\ell+1}$, and $j=0$.
\State
Find $Q^{a_0}$ by performing a walk at level $\ell$ starting at $Q^{\ell+1}$. 
\While{$\df(a_{j}, C) > \df(a_{j}, Q^{a_j})$}
\If {$\df(v,Q^{a_j}) > \df(v,\newc)$}           
\State Find $a_{j+1} \in a_jv$ 
by performing a segment query for $a_jv$ on $\FVD(Q^{a_j})$. 
\State Find $Q^{a_{j+1}} \in F^{(\ell)}$ by a walk at level $\ell$  starting at $Q^{a_j}$. 
\State Set $j = j+1$.
\Else \State Exit and \Return nil.
\EndIf
\EndWhile
\State Set $Q^\ell = Q^{a_j}$ and $u^\ell = a_j$. 
\EndFor
\State Exit and \Return $Q^0$, $u^0$.
  \end{algorithmic}
\end{algorithm}

\begin{lemma}
\label{lemma:segm_query_hvd}
The expected length of the sequence $(a_j)_{j=0}^r$ at one level of the hierarchy is $O(1)$. 
\end{lemma}

\begin{proof}
\label{app:par-pl-hvd}
Consider the sequence $(a_j)$ at level $\ell$. Let $a = a_0$, and let 
$P$ be the cluster at level $\ell+1$ that is nearest to $a$.
We first prove that for each $j = 0,\dots,r-1$,  $a$ is closer to $Q^{a_j}$ than to $P$. 
By the construction of the sequence, 
$\df(a_{j+1}, C) = \df(a_{j+1}, Q^{a_{j}})$ and  $\df(v, Q^{a_{j}}) > \df(v, \newc)$. 
Since clusters $C$ and $Q^{a_j}$ are non-crossing, cluster $Q^{a_{j}}$ is   either
forward or rear limiting for $C$ with respect
 to point $a_{j+1}$ (see Definition~\ref{def:limiting}).
 By Property~\ref{prop:limiting}, $a$ is closer to $Q^{a_{j}}$ than to $C$. 
 Since $\df(a,C) = \df(a,P)$, we have  $\df(a,Q^{a_j}) < \df(a,P)$. 
Similarly to the proof of Lemma~\ref{lemma:walkconstexp},
we can derive that the expected number of clusters in $F^{(\ell)}$ that are closer to $a$ than to $P$ is constant. 
In addition, clusters $Q^{a_j}$ for each $j = 0,\ldots, r-1$, are
distinct. Thus, $r$ is expected $O(1)$, which proves the claim. 
\end{proof}

\begin{lemma}
\label{lemma:par-pl-time}
Parametric point location in the Hausdorff Voronoi hierarchy
can be performed in expected $O(\log{n}\log{k})$ time.
\end{lemma}
\begin{proof}
The expected number of clusters at level $h$ of the Voronoi hierarchy
is O(1)~\cite{KARAVELAS-vd_of_co_in_the_plane}, 
 and 
computing the  distance from a point to a cluster requires
$O(\log{n})$ time; 
thus  Line~1 of  Procedure~\ref{alg:par-pl} requires expected $O(\log{n})$ time. 
 At a level $\ell$, $\ell = 0,\ldots,h-1$, Procedure~\ref{alg:par-pl} identifies 
 points of the sequence $(a_j)$ one by one, each time performing a walk and a segment query. 
The expected number of such walks and segment queries 
is $O(1)$ (see Lemma~\ref{lemma:segm_query_hvd}), 
each walk performs expected $O(1)$ steps  (see Lemma~\ref{lemma:walkconstexp}),
and each step of the  walk requires $O(\log{n})$ time (see Lemma~\ref{lemma:step}).
Each segment query can be performed in time $O(\log{n})$ (see Lemma~\ref{lemma:sepdecomp}). 
Since the expected number of levels in the Voronoi hierarchy is
$O(\log{k})$ (see Lemma~\ref{lemma:vh-size}), the claim follows.
 \end{proof}
 
 \subsection{Updating the Voronoi hierarchy}
 \label{subsec:update-vh}

To insert a new cluster $C$
in the Hausdorff  Voronoi hierarchy, we traverse 
the hierarchy starting at level 0 until a randomly computed maximum
level for $C$, denoted as $\ell(C)$, is found.
Inserting  $\newc$ at a level $\ell$ may make the region of a cluster $P$ at this level empty.

\begin{definition}
 A cluster $P \in \AbsFamily$ is called 
 \emph{critical at level $\ell$} with respect to $C \not\in F$,
if $\hreg^{(\ell-1)}_{\AbsFamily}(P) \neq \emptyset$,
$\hreg^{(\ell-1)}_{\AbsFamily\cup\{C\}}(P) = \emptyset$,
and $\hreg^{(\ell)}_{\AbsFamily\cup\{C\}}(P)\neq\emptyset$.
\end{definition}

Such a critical cluster $P$ becomes an obstacle to
correct point location. 
Indeed, if a query point lies in $\hreg^{(\ell)}_{\AbsFamily\cup\{C\}}(P)$,
we do not know where to continue the point location
at level $\ell-1$. To fix the problem, $P$ 
could be deleted from all levels of the hierarchy, however, this is
computationally expensive.
Instead of deleting $P$, we link $P$ to the cluster or to the  pair of clusters responsible
for the empty region of $P$; 
one of 
the responsible clusters is $C$. 
There are the following cases:

\begin{enumerate}
\item Cluster $C$ is a killer for $P$, and $\ell(C) = \ell-1$.
\item There is a cluster $K \in \AbsFamily^{(\ell-1)}$ such that $\{C,K\}$ is a killing pair for $P$,  
and one of the following holds:
\begin{itemize}
\item[(a)] $\ell(K) \geq \ell$, and $\ell(C) = \ell-1$;
\item[(b)] $\ell(C) \geq \ell$, and $\ell(K) = \ell-1$;
\item[(c)] $\ell(C) = \ell(K)  = \ell-1$.
\end{itemize}
\end{enumerate}

In cases 1 and 2(a), we link cluster $P$ to cluster $C$ only (see Lemma~\ref{lemma:linking-easy})
In case 2(b) we link $P$ to cluster $K$, and in case 2(c) to both
clusters $C$ and $K$. 
In the latter two cases we also need to identify  cluster $K$. 
The linking process is detailed in Procedure~\ref{alg:linking}.

\begin{lemma}
\label{lemma:linking-easy}
Cases~1 or~2(a) occur if and only if all the $P$-mixed vertices 
on the boundary of $\hregl(P)$
are closer to $C$ than to $P$.
\end{lemma}
\begin{proof}
Suppose we have cases 1 or 2(a), i.e., 
$C \in F^{(\ell-1)}$, $C \not\in F^{(\ell)}$, and either $C$ is a killer for 
$P$ (case~1) or there is $K \in F^{(\ell)}$ such that $C, K$ is a killing pair for $P$ (case~2(a)). 
In any of these two cases the addition of $C$ to $\HVD(\AbsFamily^{(\ell)})$  would make 
the region of $P$ empty ($\hreg_{F^{(\ell)}\cup C}(P) = \emptyset$). 
Thus, each point in ${\hreglc}(P)$, including all the $P$-mixed vertices,
is closer to $C$ than to $P$.

Now suppose that all $P$-mixed vertices 
on the boundary of $\hregl(P)$
are closer to $C$ than to $P$. Clearly none of these $P$-mixed vertices is contained in
 $\overline{\hreg}_{\AbsFamily^{(\ell)}\cup{C}}(P)$.
We need to prove that either $C$ is a
killer for $P$ (case 1) or $C,K$ is the killing pair for $P$ for some cluster $K \in F^{(\ell)}$. 
Suppose on the contrary,  that neither of these two cases holds.
Then by Property~\ref{prop:killingpair}, $\hreg_{\AbsFamily^{(\ell)}\cup{C}}(P)$ is not empty, 
and by Property~\ref{prop:connectivity}a,    $\hreg_{\AbsFamily^{(\ell)}\cup{C}}(P)$ has at least two 
$P$-mixed vertices on its boundary. 
These vertices are equidistant to $P$ and $C$; let $v$ be any of them.
Since $P$ and $C$ are non-crossing, 
we have that $C$ is forward or rear limiting for $P$ with respect to $v$ (see Definition~\ref{def:limiting}).
By Property~\ref{prop:limiting},  there is a subtree of $\fskel(P)$ 
incident to $v$ ($\fskel_v^r$ or $\fskel_v^f$) such that all its points are closer to $P$ than to $C$. 
This subtree includes 
at least one $P$-mixed vertex on the boundary of the region of $P$ in  $\HVD(F^{(\ell)}\cup{C})$;
a contradiction. 
\end{proof}

\begin{algorithm} 
 \caption{
Linking cluster $P$ that is critical at level  w.r.t. cluster $C$}
   \label{alg:linking}
  \begin{algorithmic}[1]
  \State Let  $V_{\ell-1}$ be the list of the $P$-mixed vertices
on $\bd \hreg_F^{(\ell-1)}(P)$.

\State Let  $V_{\ell}$ be the list of the $P$-mixed vertices 
on  $\bd \hreg_F^{(\ell)}(P)$.

    \If {all vertices in $V_{\ell}$ are closer to $C$ than to $P$} 
    \State Link  $P$ to $C$ and \Return 
  \Else  
    \State Let $v \in V_{\ell}$ be closer to $P$ than to $C$. 
    \State Let $c \in C$ be such that  $\df(v,\newc)=d(v,c)$.
   \ForAll {$u \in V_{\ell-1}$}
   \State let $p_1, p_2 \in P$ and $q \in Q$ be such that
   \State $v$ borders $\hregl(p_1)$, $\hregl(p_2)$ and $\hregl(q)$.
   \If{$c$ and $q$ lie on different sides of chord $\overline{p_1p_2}$}
   \State set $K = Q$.
   \EndIf
   \EndFor
   \If {$\ell(C) \geq \ell$}
   \State  Link $P$ to $K$ and \Return
   \Else \State{ Link $P$ to $\{C,K\}$ and \Return}
   \EndIf
   \EndIf
  \end{algorithmic} 
\end{algorithm}

\begin{lemma}
 \label{lemma:linking-correct}
Procedure~\ref{alg:linking} performs the linking correctly.
That is,
for any point $x \in \hreg_{F \cup C}^{(\ell)}(P)$, $\df(x,Q) <
\df(x,P)$, where $Q$ is the cluster (or
one of the two clusters) linked to $P$
by Procedure~\ref{alg:linking}.
\end{lemma}
\begin{proof}
We need to prove that Procedure~\ref{alg:linking} always identifies a cluster 
(or a pair of clusters)
such that any point in $\hregl(P)$ is closer to these cluster(s)
than to $P$.

Suppose that Line~4 of the procedure is executed i.e., linking is done to cluster $C$. 
Then by Lemma~\ref{lemma:linking-easy}, cases~1 or~2(a) occur, and
Property~\ref{prop:killingpair} guarantees that any point in $\hregl(P)$ is closer to 
$C$ than to $P$. 
Thus, the linking to $C$  is correctly done.

Suppose that Procedure~\ref{alg:linking} does not terminate at Line~4. Let 
 vertex $v$ and point $c$ be as determined in Lines~6 and~7 respectively.
Since $\df(v,C)> \df(v,P)$ and $\df(v,C) = d(v,c)$, we have 
 $c \not\in \conv{P}$.
Since  $\hreg^{(\ell-1)}_{\AbsFamily\cup\{C\}}(P) = \emptyset$
and $v$ is closer to $P$ than to $C$,
we have $v \not\in \hreg^{(\ell-1)}_{\AbsFamily}(P)$.
 Cluster $\newc$ is (forward or rear) limiting for $P$  with 
 respect to any $P$-mixed vertex $w$
 on the boundary of $\hreg_\AbsFamily^{(\ell-1)}(P)$. Suppose, without loss of generality, 
 that $\newc$ is forward limiting,
i.e., $\newc \subset  \cldf_w \cup \conv{P}$;
then $c \in \cldf_w \setminus \conv{P}$.
 Let $u$ be the first $P$-mixed vertex
 encountered as we traverse $\fskel(P)$ from $v$ to its portion
 enclosed in $\hreg_{\AbsFamily}^{(\ell-1)}(P)$.
 Let $Q$ be the cluster inducing $u$, and $q$ be the point in $Q$ such
 that $d(u,q)=\df(u,Q)$.
The pair $\{Q,\newc\}$ is by definition a killing pair for $P$, and $q$, $c$
lie at opposite
sides of the chord of $P$ inducing $u$.
All other $P$-mixed vertices 
$v_i$
on $\bd \hreg_\AbsFamily^{(\ell-1)}(P)$, 
$v_i\neq u$, 
must be induced by forward limiting clusters (see Property~\ref{prop:limiting}). 
Thus, any point $q_i$, $q_i\neq q$, inducing a $P$-mixed vertex $v_i$ 
must lie on the same
side of the corresponding chord as $c$. Thus, Line~12 correctly sets
$K=Q$.
In Line~13, the condition 
distinguishes between cases~2(b) (Line~14) and~2(c) (Line~16). 
Property~\ref{prop:killingpair} again guarantees 
the correctness of Lines~13-16.
\end{proof}

We summarize 
in the following theorem. 

\begin{theorem}
\label{lemma:vh}
The Voronoi hierarchy for the Hausdorff Voronoi diagram 
of a family of $\nclus$ clusters of total complexity $n$
has expected size $O(n)$. 
Both the point location query and the 
parametric point location query
can be performed in expected time $O(\log{n}\log{\nclus})$. 
Insertion of a cluster takes 
$O((N/{\nclus})\log{n})$ amortized  time, where $N$ is the 
total number of update operations
in all levels during the insertion of all $k$ clusters.
\end{theorem}

\begin{proof}
  The expected space of the Voronoi hierarchy
  is analyzed in Lemma~\ref{lemma:vh-size}.
  Lemmas~\ref{lemma:vh-size} to~\ref{lemma:d-reduced} imply that  
  point location in the Voronoi hierarchy
can be done in expected
  $O(\log{n}\log{\nclus})$ time.
By Lemma~\ref{lemma:par-pl-time}, parametric point location is
performed in 
expected $O(\log{n}\log{k})$ time.
  During the insertion of a cluster,  
  two procedures are performed: updating the diagram at all
 necessary levels, and the linking of regions that disappear.
Since updating each (constant-sized) element of a diagram 
requires $O(1)$ time, the total time required for all update operations to insert all $k$ clusters is 
 $O(N)$. 

 Consider the linking of a cluster $P$ that is 
 critical at level $\ell$ with respect to a cluster $C$.
  We visit the $P$-mixed vertices of $\HVD(\AbsFamily^{(\ell-1)})$, 
but these vertices get  deleted during the same step. We also visit 
the $P$-mixed vertices of $\HVD((\AbsFamily \cup\{ C\})^{(\ell)})$. 
The latter vertices are visited 
  at most twice: 
  when $P$ is critical at level $\ell+1$ 
  and when $P$ is critical at level $\ell$. 
  The  time complexity of each of these visits is $O(\log{n})$.
  Thus, the claimed complexity 
bound
follows.
\end{proof}

\section{Tracing a new Voronoi region}
\label{sec:tracing}

In this section we give details on   
how to compute 
the boundary of a new region
$\hreg_{F_{i}}(C_{i})$  within  $\HVD(F_{i-1})$, given a
representative point $\rep$ in $\hreg_{F_{i}}(C_{i})$.
The first task is to determine a point $w$
on the boundary of $\hreg_{F_{i}}(C_{i})$. Then tracing can be
performed as described in~\cite{EP-HVD-d_and_q}.

We maintain a {refinement} of $\HVD(F)$, 
called the \emph{visibility-based  decomposition}~\cite{EP-HVD-d_and_q}, 
which is denoted by $\HVD^*(F)$ and is obtained as follows:
For every Voronoi vertex $x$ on
the Hausdorff boundary of $\hreg(c)$, $c\in C$,
add to the diagram the segment
$s = {cx} \cap \hreg(c)$, see Fig.~\ref{fig:vb-trace}a. 
A face $f$ of $\HVD^*(F)$ consists of four sides;
one side is a chain of 
the farthest boundary, called the 
\emph{$\fskel$-chain} of $f$, see Fig.~\ref{fig:vb-trace}b.
The \emph{$\fskel$-chain} of $f$ may have non-constant complexity,
however, the other three sides of $f$ consist of at most one edge each.

\deleted{
Once the representative point $\rep$ is identified the tracing of
$\hreg_{F_{i}}(C_{i})$ is done similarly to
\cite{EP-HVD-d_and_q} in time proportional to the complexity of the
new regions and the updates of the diagram. More details are given
in  Section ~\ref{sec:tracing}.
}

The main algorithm in Section~\ref{sec:algorithm} has identified  
a segment $\rep v$, along an edge of  $\fskel(C_i)$,
where $\rep$ is the representative point in $\hreg_{F_{i}}(C_{i})$,  and $v$ is
the parent of $\rep$ in $\fskel(C_i)$ such that  
$v \notin \hreg_{F_{i}}(C_{i})$. 
We determine a $C_i$-mixed vertex $w$ along $\rep v$.
To this goal, we trace segment $\rep v$ through $\HVD(F_{i-1})$,
starting at $\rep$, until we determine $w$.

\begin{figure}
  \begin{minipage}{0.49\linewidth}
  \centering
 \includegraphics {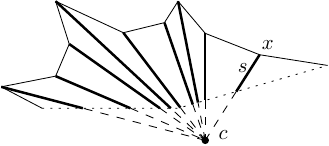}
\\
(a)
  \end{minipage}
 \hfill
   \begin{minipage}{0.49\linewidth}
   \centering
 \includegraphics {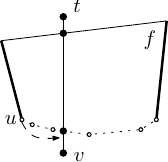}
\\
(b)
  \end{minipage}
  \caption{(a) Visibility-based decomposition of $\hreg_C(c)$; 
  (b) Tracing the $\fskel$-chain of face $f$ (dotted lines), starting from the endpoint $u$.}
  \label{fig:vb-trace}
\end{figure}

In more detail, let  $f$ be a face of $\HVD^*(F_{i-1})$ intersected by $\rep v$. 
Initially, $f$ is the face of $\HVD^*(F_{i-1})$ containing $\rep$. 
In constant time, we check whether 
$w$ lies in the interior of $f$,  and if so we identify $w$.
If it does not, we move to the face $g$ of $\HVD^*(F_{i-1})$ that is adjacent to $f$ and is  
intersected by 
segment
$\rep v$. To identify $g$,  we may need to trace a
portion of the \emph{$\fskel$-chain} of $f$. This is performed as follows:
Among the two endpoints of  the $\fskel$-chain, at least one must be in 
 $\hreg_{F_{i}}(C_{i})$, by Property~\ref{prop:connectivity}a. 
We first identify such an endpoint $u$, and then we trace the
$\fskel$-chain, starting at $u$, until we meet its intersection with $\rep v$, 
see Fig.~\ref{fig:vb-trace}b. 
At this time,  we have determined $g$ and we can continue our search for $w$ with
$f=g$. Tracing the $\fskel$-chain of $f$, starting at $u$, has no effect on
the overall time complexity
because all  traced edges of the $\fskel$-chain intersect
$\hreg_{F_{i}}(C_{i})$,
and thus, they are guaranteed to be deleted from $\HVD(F_{i-1})$ by
the main algorithm during the insertion of $C_i$ at step $i$.
To identify $u$, we consider both endpoints of the $\fskel$-chain,  and 
compare their distances to $C_i$ and  to their closest cluster in $F_{i-1}$.
The latter distance  is readily available from $\HVD(F_{i-1})$.
To derive the former distance, we perform point location in $\FVD(C_i)$.  
Thus, in the worst case, we perform two point locations in $O(\log n)$
time, and 
we trace 
a number of edges of  $\HVD(F_{i-1})$, none of which will appear in $\HVD(F_i)$, 
spending $O(1)$ time per edge.

After $w$ is identified, the 
tracing of the boundary of
$\hreg_{F_{i}}(C_{i})$ is performed in time proportional to  the total
number of edges that are inserted or deleted from the Hausdorff
diagram during step $i$, plus $|C_i|$. 
Note that to identify the new Voronoi vertices we simply walk sequentially along
edges of $\HVD(F_{i-1})$ and $\FVD(C_i)$, which are deleted, using the
visibility-based decomposition \cite{EP-HVD-d_and_q}.
To identify $w$, we also perform point location, thus, an $O(\log n)$ factor is
multiplied to the above quantity.

We conclude that the  time complexity for tracing $\hreg_{F_{i}}(C_{i})$ is proportional to the
number of updates (insertions and deletions) in the Hausdorff diagram as a
result of inserting cluster $C_i$, multiplied by $O(\log{n})$. 
Combining with the overall time complexity analysis, given in the following section,
the total expected time devoted to the
tracing of new regions throughout the algorithm is $O(n\log n)$.

\section{Complexity analysis} 
\label{sec:compl_analysis}

 The running time of our algorithm
  depends on the number of update operations (insertions and deletions)
  during the construction of the diagram.
Based on the Clarkson-Shor technique~\cite{Clarkson_rand_sampling_2},
 we prove that the expectation of this number   is linear, when
 clusters are inserted in random order.
In the Hausdorff Voronoi diagram, sites (clusters) do not have constant size, 
as it is typically assumed in the literature.
Thus, we need to adapt the standard probabilistic arguments in this environment.

\begin{theorem}
\label{thm:numberofoperations}
Given a family $\AbsFamily$ of non-crossing clusters of points, 
the expected total number of update
operations during the randomized incremental construction of
$\hvd(\AbsFamily)$ is $O(n)$, 
 where $n$ is the total complexity of the clusters in $\AbsFamily$.
 \end{theorem}

Theorem~\ref{thm:numberofoperations} can be  extended 
to all levels of the Voronoi hierarchy as stated in the following
corollary. We defer its proof to Section~\ref{subsec:proof-vh-operations}, 
after the proof of Theorem~\ref{thm:numberofoperations}.

\begin{corollary}
\label{lemma:expectedupdatesVH}
The expected number of update operations made on all the levels of the Hausdorff Voronoi hierarchy of $\InpSet$
during the incremental construction is $O(n)$.
\end{corollary}

We conclude with the following theorem.

 \begin{theorem}
 \label{thm:voronoihierarchy}
 The Hausdorff Voronoi diagram of a family $F$ of non-crossing clusters can be
 constructed in
 $O(n\log{n}\log{\nclus})$
 expected time and $O(n)$ expected
 space, where $k$ is the number of clusters in $F$ and $n$ is number of points in all clusters.
 \end{theorem}

 \begin{proof}
  As a preprocessing, we build the 
  centroid decomposition
  for each cluster in $F$, in total $O(n\log n)$ time 
  (see Lemma~\ref{lemma:sepdecomp}).
The algorithm to insert a cluster $C \in F$ does the following:
(1) searches for a representative point in the new Hausdorff Voronoi region;
(2) traces the boundary of the new region
 (see Section~\ref{sec:algorithm}); and~(3) inserts $C$ in
 the Voronoi hierarchy (see Section~\ref{sec:vh}).
 By the
discussion in Section~\ref{sec:tracing}, and by 
Theorem~\ref{thm:numberofoperations} and Corollary~\ref{lemma:expectedupdatesVH}, 
 the total time to perform~(2) and~(3), for all clusters, is expected $O(n\log{n})$.
Searching for a representative point in the Hausdorff
Voronoi region of $C$ (part~(1)) performs  $O(\cardinality{C})$ point
location queries  and at most one parametric point location query
in the Voronoi hierarchy.
Combining with Lemma~\ref{lemma:sepdecomp} and Theorem~\ref{lemma:vh},
we derive that the total expected time to determine
a representative point for all clusters is  $O(n\log{n}\log{\nclus})$;
the claim follows. 
 \end{proof}

\begin{remark}
Deterministic $O(n)$ space complexity can  be achieved by using a
dynamic point location data structure for a planar subdivision
\cite{Arge,Baumgarten}. 
On this data structure,
parametric point location can be performed
as described
in Cheong~{et al.}~\cite{fpvd2011cg}. The time complexity of such a query 
is $t_q^2$, where $t_q$ is the time complexity of  point location in the 
chosen data structure. 
In particular, the data structure by Baumgarten~{et al.}~\cite{Baumgarten}
has $t_q \in O(\log n \log\log n)$, which leads to the 
construction of the Hausdorff Voronoi diagram with expected running time 
$O(n\log^2{n}(\log\log{n})^2)$ and deterministic space $O(n)$.
\end{remark}

\subsection{Proof of Theorem~\ref{thm:numberofoperations}} 
\label{subsec:operations}

In order to count the number of update operations of the
algorithm, 
i.e., insertions and deletions of \emph{features} 
such as vertices, edges, and faces  
of the
incrementally constructed diagram, 
we will associate 
each update operation with a feature of the diagram.
%
Each
feature 
has been inserted by an operation. If a feature
is deleted, then it 
cannot be inserted again in the future.
As a result, the number of deletion
operations is bounded by the number of insertion operations.
%
Thus, we intend to prove that the expected number of features that
appear during the construction of the diagram is $O(n)$.
To 
this  end,
we can ignore features 
that are 
associated only with the farthest Voronoi
diagram  of each cluster, because their total worst case combinatorial
complexity is $O(n)$.

\paragraph{Configurations.}

We give some definitions related to  features of the diagram.

\begin{definition}
\label{def:configuration}
A \emph{configuration} is a triple of points $(p,q,r)$ such that $p$,
$q$, $r$ lie on the boundary of a disk $D$ and $q$ is contained in the
interior of the counterclockwise arc from $p$ to $r$.
We call $D$ the \emph{disk of the configuration}, its center
the \emph{center} of the configuration, and the counterclockwise arc
$pr$ the \emph{arc of the configuration}.
See Fig.~\ref{fig:configuration}.

A configuration is \emph{pure} if its three points
belong to \emph{three} different clusters of $F$
and all other points of these three clusters are contained in the
interior of the disk of the configuration.

A configuration is \emph{mixed} if its three points
belong to \emph{two} different clusters of $F$
and all other points of these two clusters are contained in the
interior of the disk of the configuration.
\end{definition}
From now on, configurations of our interest will be either pure or
mixed. Therefore, each configuration is either associated with three
(a pure one) or two (a mixed one) clusters.

\begin{figure}
\centering
\includegraphics{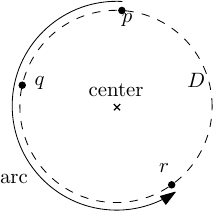}
\caption{A configuration $(p,q,r)$}
\label{fig:configuration}
\end{figure}

\begin{definition}
\label{def:conflict}
A cluster $C$ is in \emph{conflict with a configuration}
if
(a)~$C$~does not contain any of the points in the configuration,
and
(b)~$C$~is contained in the union of
the interior of the disk of the configuration
and the arc of the configuration.

The \emph{weight} of a configuration is the number of clusters
in conflict with it.
\end{definition}

Definition~\ref{def:conflict} is general and it does not follow the general position
assumption stated in Section~\ref{sec:prelim}.
Under this assumption (b) can simplify to: 
``(b)~$C$ is contained in the union of the
interior of the disk of the configuration''.

\begin{lemma}
The number of zero-weight configurations of $F$ is of the same order
as the combinatorial complexity of the Hausdorff Voronoi diagram of
$F$.
\end{lemma}
\begin{proof}
Each zero-weight configuration is associated with a vertex of the
Hausdorff Voronoi diagram.
Indeed, the center of this configuration is at the vertex and the disk
of the configuration
that contains the clusters associated with the
configuration.
Consider a vertex $v$ of the Hausdorff Voronoi diagram. The degree of
$v$ in the diagram
equals the number of configurations with center $v$ plus the number
of some features that are associated just with farthest Voronoi
diagrams (that we have claimed before that we can ignore).
As a result, zero-weight configurations estimate well the combinatorial
complexity of the Hausdorff Voronoi diagram.
\end{proof}

\paragraph{Configurations of weight at most $k$.}

Let $\Kpure_0(F)$, $\Kpure_k(F)$, $\Kpure_{\leq k}(F)$ denote the sets of
pure configurations of zero weight, weight equal to $k$, and weight at most
$k$, of a family $F$ of non-crossing clusters, respectively.
Let $\Npure_0(F)$, $\Npure_k(F)$, $\Npure_{\leq k}(F)$
denote the cardinality of the
aforementioned sets, respectively.
Define analogously the sets of mixed configurations
$\Kmix_0(F)$, $\Kmix_k(F)$, $\Kmix_{\leq k}(F)$ and their
cardinalities
$\Nmix_0(F)$, $\Nmix_k(F)$, $\Nmix_{\leq k}(F)$, respectively.
Both $\Npure_0(F)$ and
$\Nmix_0(F)$ are 
\(  O\bigl(\sum_{C \in F}\cardinality{C}\bigr) = O(n) \)~\cite{ep2004algorithmica}.
Then, using the Clarkson-Shor
technique~\cite{Clarkson_rand_sampling_2},
and  in particular
~\cite[Theorem~1.2]{Sharir}, with a random sample of the clusters in $F$, we  obtain:

\begin{equation*}
\Npure_{\leq k}(F) \leq \cpure \cdot n k^2
\text{ and }
\Nmix_{\leq k}(F)  \leq \cmix \cdot n k  ,
\end{equation*}
for $k>0$ and some constants $\cpure$ and $\cmix$.
The details to obtain 
these bounds are quite standard and we
refer the interested reader to \cite{Clarkson_rand_sampling_2,Sharir}.

\paragraph{Appearance of a feature.}

Consider a configuration $c$ of weight $k$ in family $F$ with $m$
clusters.\footnote{Note the difference in notation from previous sections: 
here $k$ denotes the weight of a configuration and $m$ denotes the number of clusters in F.
We do this change in order to be consistent with the notation of Sharir~\cite{Sharir}.} 
Assume the Hausdorff Voronoi diagram of $F$ is constructed with the
incremental algorithm and the clusters are inserted according to
permutation $\pi$.
The feature corresponding to $c$ appears at some stage of the
incremental algorithm if and only if
the clusters associated with $c$ occur in $\pi$ \emph{before} the
$k$ clusters that conflict with configuration $c$.
This event happens with probability
\[\Pr[\text{pure $c$ feature appears}] = \frac{3!k!}{(k+3)!} =
\frac{6}{(k+1)(k+2)(k+3)}\]
for pure configurations, and with probability
\[\Pr[\text{mixed $c$ feature appears}] = \frac{2!k!}{(k+2)!} =
\frac{2}{(k+1)(k+2)}\]
for mixed configurations.

The expected number of appearances of features corresponding to a pure
configuration is therefore:
\begin{align*}
& \sum_{k=0}^{m-3}  \sum_{c \in \Kpure_k(F)} \Pr[\text{pure $c$ feature appears}]
= \sum_{k=0}^{m-3} \sum_{c \in \Kpure_k(F)}
\frac{6}{(k+1)(k+2)(k+3)}
\\
& \quad {}  = 6 \sum_{k=0}^{m-3} \frac{\Npure_k(F)}{(k+1)(k+2)(k+3)}
=
\Npure_0(F) +
6 \sum_{k=1}^{m-3}
\frac{\Npure_{\leq{k}}(F) - \Npure_{\leq{k-1}}(F)}{(k+1)(k+2)(k+3)}
\displaybreak[0]\\
& \quad {}  =
\frac{3}{4}\Npure_0(F) +
18 \sum_{k=1}^{m-4} \frac{\Npure_{\leq{k}}(F)}{(k+1)(k+2)(k+3)(k+4)} +
\frac{\Npure_{\leq{m-3}}(F)}{(m-2)(m-1)m}
\displaybreak[0]\\
& \quad {} \leq
\frac{3}{4}\Npure_0(F) +
18 \sum_{k=1}^{m-4} \frac{\cpure \cdot n k^2}{(k+1)(k+2)(k+3)(k+4)} +
\frac{\cpure\cdot n(m-3)^2}{(m-2)(m-1)m}
\displaybreak[0]\\
& \quad {} \leq
\frac{3}{4}\Npure_0(F) +
18 \cdot \cpure \cdot n \sum_{k=1}^{m-4} \frac{1}{k^2} +
\frac{\cpure\cdot n}{m}
=
O(n)
\end{align*}

Similarly, the expected number of appearances of features
corresponding to a mixed configuration is:
\begin{align*}
& \sum_{k=0}^{m-2}  \sum_{c \in \Kmix_k(F)} \Pr[\text{mixed $c$ feature appears}]
= \sum_{k=0}^{m-2} \sum_{c \in \Kmix_k(F)}
\frac{2}{(k+1)(k+2)}
\\
& \quad {}  = 2 \sum_{k=0}^{m-2} \frac{\Nmix_k(F)}{(k+1)(k+2)}
=
\Nmix_0(F) +
2 \sum_{k=1}^{m-2}
\frac{\Nmix_{\leq{k}}(F) - \Nmix_{\leq{k-1}}(F)}{(k+1)(k+2)}
\displaybreak[0]\\
& \quad {}  =
\frac{1}{2}\Nmix_0(F) +
 4 \sum_{k=1}^{m-3} \frac{\Nmix_{\leq{k}}(F)}{(k+1)(k+2)(k+3)} +
\frac{\cmix\cdot n(m-2)}{(m-1)m}
\displaybreak[0]\\
 & \quad {} \leq
\frac{1}{2}\Nmix_0(F) +
4 \cdot \cmix \cdot n \sum_{k=1}^{m-3} \frac{1}{k^2} +
\frac{\cmix\cdot n}{m}
=
O(n)
\end{align*}

 Therefore, we have proved the following, which implies
Theorem~\ref{thm:numberofoperations}.

\begin{lemma}
\label{lemma:expectedupdates}
The expected number of features that appear during the incremental
construction is $O(n)$.
\end{lemma}

\subsection{Proof of Corollary~\ref{lemma:expectedupdatesVH}}
\label{subsec:proof-vh-operations}

	By the
discussion in the proof of Theorem~\ref{thm:numberofoperations},  the expected number of
structural changes 
during the incremental construction
is proportional to the expected number of appearing features	 (i.e.,  pure and mixed vertices).
   For a fixed level $\ell$, the expected total number of points in $F^{(\ell)}$ is $\beta^\ell n$.
		  By Lemma~\ref{lemma:expectedupdates},
		  the expected number of features that appear during the incremental
		  construction of $\HVD(F^{(\ell)})$ at level $\ell$ is
		   $O(\beta^\ell n)$.
		  Therefore, the expected total number of features that appear at  all the levels
		  is at most $\sum_{\ell=0}^{\infty}O(\beta^\ell n) = O(n)$.

\section{Discussion and Open Problems}
\label{sec:discus}
We have provided 
algorithms of improved complexity  for constructing the
Hausdorff Voronoi diagram of a family of non-crossing
clusters of points. 
These algorithms are based on randomized incremental construction and point
location. 
There is still a gap in the complexity of constructing the
Hausdorff Voronoi diagram between our
$O(n\log^2{n})$ expected-time 
algorithm and the well-known
$\Omega(n\log{n})$ time lower bound.
An open problem is to close 
or reduce this gap.
It is interesting that in the
$L_{\infty}$ metric a simple  $O(n\log{n})$-time $O(n)$-space
algorithm is known~\cite{PX15}, 
which is 
based on a two-phase plane sweep.
In future research, 
 we plan to 
consider
families of arbitrary clusters that may be crossing. 
In this case, the size of the diagram can vary
from linear to quadratic, and therefore, an output-sensitive 
approach would be 
 most desirable. 

\section*{Acknowledgements}
We thank an anonymous reviewer for valuable comments that helped improve the presentation of this paper.


\bibliographystyle{spmpsci}     
\bibliography{hvd-bibl}
\end{document}